\documentclass[reprint,amsmath,amssymb]{revtex4-2}
\usepackage[letterpaper, portrait, margin=1in]{geometry}
\usepackage[utf8]{inputenc}
\usepackage{graphicx, caption}
\usepackage{amsfonts}
\usepackage{amsmath,bm}
\usepackage{IEEEtrantools}
\usepackage{amssymb}
\usepackage{amsthm}
\usepackage{color}
\usepackage{comment}
\usepackage{mathtools}
\usepackage{braket}
\usepackage{enumerate}
\usepackage{multirow}
\usepackage{hyperref}
\usepackage{balance}
\usepackage{flushend}
\usepackage{tikz} 
\usetikzlibrary{quantikz}

\theoremstyle{plain}
\newtheorem{thm}{Theorem}

\newtheorem{defn}[thm]{Definition} 

\newtheorem{lemma}[thm]{Lemma}
\newtheorem{corollary}[thm]{Corollary}

\newtheorem{example}[thm]{Example}
\newtheorem{remark}[thm]{Remark}

\begin{document}

\title{Computational reach of machine learning models is independent of untrained quantum resources}

\author{Steven Herbert}
\email{sjh227@cam.ac.uk}
\affiliation{Department of Computer Science and Technology, University of Cambridge, UK}

 \begin{abstract}
    \noindent 
This paper\footnote{This work was conducted when the author was at Quantinuum.} proves that if an untrained quantum circuit is used as a resource in a machine learning model, then on average no resource quantum circuit is better than any other that can achieve the same set of computational effects. The paper also proves a supporting theorem that even if certain mild idealisations are omitted, the average quantum advantage remains negligible at best. These results cast serious doubt on several proposals to use untrained quantum circuits in machine learning models: at best such claims should be substantiated empirically, as this paper proves there is no \textit{a priori} theoretical reason to suppose that introducing an untrained quantum circuit will increase performance.
\end{abstract}

\maketitle

\section{Introduction}
\label{sec:intro}

Over the past several years, various so-called ``quantum supremacy'' (also referred to as ``quantum advantage'') experiments have shown that, beyond reasonable doubt, universal quantum computers can achieve computational effects that are out of reach of classical computers in any reasonable time \cite{Google-supreme1, ZHU2022240, liu2025certified}. These experiments typically use random quantum circuits to sample probability distributions that are hard to prepare by classical means and, as such, they are far removed from useful applications. An early hope for \textit{useful} near-term quantum advantage was to instead train quantum circuits to do some particular task, thus delivering quantum advantage that is, by design, useful \cite{peruzzo2014vqe, mcclean2016theory, schuld2014quest, benedetti2019parameterized}. However, there is mounting evidence that the training is itself a sufficiently computationally-intensive procedure that the quantum advantage will be negated when the full cost is accounted for \cite{McClean2018,Bittel2021}.

These two facts -- that quantum computers \textit{can} perform tasks that are out of reach of classical computers, but that training circuits for particular practically-interesting problems is computationally hard -- have in turn spawned the idea that \textit{untrained} quantum circuits may be useful as resources in otherwise classical machine learning models. In \textit{quantum reservoir computing} \cite{fujii2017quantumreservoir, Fujii2021} the input is passed through an untrained quantum circuit, whilst \textit{quantum random kitchen sinks} \cite{wilson2020quantum} is a similar idea. The central idea of these papers is that there is an innate advantage in mapping classical data into a high dimensional Hilbert space, even if that is via an untrained quantum circuit. In a similar vein, the hope offered in Ref.~\cite{OrcaBlog} is that ``since the distributions from which near-term quantum computers can efficiently sample are particularly rich and complex, they are promising candidates for representing latent spaces within generative models''.

Whilst these methods may be empirically shown to lead to performance improvements for certain specific problems \footnote{Or when comparing only to a restricted classical model: for instance \cite{Cimini2026LargeScale} shows evidence that if one insists on having an optical reservoir, then there is benefit in having quantum over classical light.}, this paper emphatically refutes that any untrained quantum circuit used as a resource offers an \textit{a priori} advantage that is independent of the problem. In particular, the main results of Section~\ref{sec:iii}, Theorem~\ref{nfl1} and its corollaries, prove that if one treats \textit{untrained} as meaning Haar random, then no resource circuit is better on average than any other that can achieve the same set of computational effects, almost surely. Furthermore, Section~\ref{sec:supp-analysis} drops some simple idealisations required for Theorem~\ref{nfl1}, and shows that even then the space of computational effects is simply too large for any untrained quantum circuit to achieve anything more than (at best) a negligible advantage on average.

\subsection{Resource theories}

It is first beneficial to introduce the concept of resource theories, and for simplicity of exposition, it is further helpful to consider a simplified setting where the untrained quantum circuit provides a fixed quantum state that is measured and taken as the latent distribution of a generative model. We then ask whether it is the case that samples obtained by measuring a certain quantum state are more useful on average? This question is aptly formalised using resource theory \cite{resource1, resource2}: the task of generative modelling may be summarised as the task of converting samples from the latent distribution (i.e., some \textit{resource} randomness) into samples from the target distribution. We thus ask whether certain latent distributions reduce the average cost of preparing the family of target distributions, as this would mean such distributions have greater resource value (in this context).

According to the literature there is reason to suppose directly contradictory things about what would make a good resource for this purpose. On the one hand, the fact that there exist quantum states that sample classically hard-to-prepare distributions when measured suggests that complex highly-entangled states may have higher resource value (exactly the contention of Ref.~\cite{OrcaBlog}). On the other hand, the fact that a resource theory of \textit{uncomplexity} has recently been proved for quantum states in general \cite{YungerHalpern2022} hints simple, unentangled states may have higher resource value. (It should be noted that the resource theory of uncomplexity addresses a different question to ours: it is concerned with the resource value of quantum states that may then undergo further unitary evolution -- or in fact slightly noisy evolution -- whereas we are concerned with those states being measured to create classical randomness.) 

With these two incompatible precedents in place, it is perhaps unsurprising that neither is true -- at least in the most general sense -- as we prove in this paper.

\subsection{Analogy with \textit{no free lunch} theorems}

\textit{There ain't no such thing as a free lunch} is an old adage that captures the truism that you usually pay for something indirectly, even if it seems to be free. This metaphor is often invoked in social and natural sciences as well as economics to allude to the fact that when a certain method or resource appears to be unreasonably successful, it is usually because only part of the problem space is being considered (something that may not be explicitly noted) and that this is particularly well-matched to the method / resource in question.

Though it was later connected to Hume's induction skepticism \cite{Schurz2017-SCHNFL-3}, the first (explicit) no free lunch theorem was proposed in a paper by Wolpert, and concerned the problem of choosing between learning algorithms \cite{nfl1}. The term ``no free lunch theorem'' was then popularised in a paper of Wolpert and Macready \cite{nfl2}, on the problem of choosing between optimisation algorithms. This provides a suitable exemplar for no free lunch theorems in general: when averaging over the entire class of optimisation problems, no optimisation algorithm is better or worse than any other. This gives rise to, for example, the surprising implication that if one treats optimisation as minimisation in particular, then hill-climbing algorithms (which aim to make ``uphill moves'') perform equally well on average as gradient descent algorithms. 

This counter-intuitive result can be explained to some extent by the fact that, in practice, the entire space of minimisation problems is not what is of interest, and in fact there is some hidden structure -- for instance that the space of problems of practical interest is heavily biased in favour of those optimisation problems with surfaces exhibiting some smoothness -- in which case gradient descent algorithms do indeed out perform hill-climbing algorithms for minimisation. Nevertheless, the central premise of a no free lunch theorem is that there is no \textit{a priori} bias over the space of problems as so the cost must be averaged over the entire space with equal weighting. 

In this paper we vary not the algorithm itself, but rather a resource available to the algorithm, but the analogy with no free lunch theorems is immediate\footnote{Note that in an earlier version of this paper -- still available, for example on arXiv, the headline theorem was in fact referred to as a ``no free lunch theorem''. However, after much deliberation this labelling was dropped, as the fact that here a resource not algorithm is varied seemed to introduce unhelpful ambiguity.}: in the absence of some ``matching'' between the problem and the resource, there is no reason to believe \textit{a priori} that a certain resource should be more beneficial than any other.

\section{A general structure for machine learning models with untrained quantum circuits as resources}
\label{sec:set-up}

A machine learning model is a map from an input, which may include uniform randomness, to an output. For instance, a generative model may take \textit{only} uniform randomness as an input and map this to a sample from the desired distribution. Alternatively, a binary classifier is a map from any suitable input variable to a single bit output. Although in modern computer science there are many levels of abstraction, on a digital (classical) computer fundamentally this map must be realised with gates from a suitable gateset, and the set $\{ \text{NAND}, \text{ FANOUT}\}$ is universal for classical computation. Our analysis treats the map as reversible, in which case the gateset $\{ X, \text{TOFFOLI} \}$ is sufficient, as this enables the reversible implementation of both NAND and FANOUT. In particular, as long as an ancilla qubit is available (which is hereafter omitted from consideration) $\{ X, \text{TOFFOLI} \}$ can generate any permutation.

In the following analysis it is convenient to treat the permutation as acting on a quantum state, and for this we establish the following simple result:

\begin{lemma}
\label{gen-defer}
Circuits of the form:

\begin{center}
\begin{minipage}{0.35\textwidth} 
    \begin{center}
    \begin{quantikz}
            \lstick{$\ket{\psi}$} & \qwbundle{}  & \meter{} & \gate{ P } \hphantom{wide} & \qw
\end{quantikz} \\
$\textnormal{(a)}$  \end{center}
    \end{minipage}
    \begin{minipage}{0.35\textwidth}
    \begin{center}
    \begin{quantikz}
            \lstick{$\ket{\psi}$} & \qwbundle{}   & \gate{ P } \hphantom{wide} & \meter{} & \qw 
\end{quantikz}  \\
$\textnormal{(b)}$ \end{center}
    \end{minipage}
\end{center}

\noindent sample the same distribution, for any given input quantum state $\ket{\psi}$ and permutation, $P$.
\end{lemma}
\begin{proof}
Consider the probability of measuring some output bitstring $x$. First addressing circuit (b), in which case the probability is $|\bra{x} P \ket{\psi}|^2$. Now consider the circuit (a), in which case the probability of sampling the bitstring $x$ is $|\braket{y | \psi}|^2$ where $\ket{y} = P^\dagger \ket{x}$. That is, because $P$ is a permutation, $x$ will be sampled if and only if the measurement outcome is $y$ such that $P\ket{y} = PP^\dagger \ket{x} = \ket{x}$. The probability of this happening is:
\begin{equation}
|\braket{y | \psi}|^2 = |(\ket{y})^\dagger \ket{\psi}|^2 = |(P^\dagger \ket{x})^\dagger \ket{\psi}|^2 = |\bra{x} P \ket{\psi}|^2
\end{equation}
and thus the probability of measuring any bitstring, $x$, is identical for the two circuits (a) and (b).
\end{proof}

With this established, there are two essential ways in which an untrained quantum circuit can act as a resource in a (otherwise classical) machine learning model. Firstly, the input could be passed through the quantum circuit, with the output measured and then fed into the classical machine learning model; secondly, the quantum circuit could just produce a static ``resource'' state that does not vary with the input and instead the same measured state is always fed into the machine learning model. To distinguish these, we henceforth refer to the former as a quantum circuit acting as a resource, and the latter as a quantum state acting as a resource. We now formally define these, where the machine learning model takes an input, $\ket{x}$ of size $n_x$ bits, as well as some $n_+$ uniformly random bits (input $\ket{+}$), and some $n_0$ workspace (qu)bits initialised in the state $\ket{0}$. The resource unitary consumes some $n_q$ bits initialised in the state $\ket{0}$ (as well as the input $\ket{x}$ if the quantum circuit is used as a resource). The output is labelled $\ket{y}$ and has size defined as $n_y$. Note that $\ket{\tilde{y}}$ is included for completeness, and is the (qu)bits not included in the output (and so may be of size zero if every output (qu)bit forms part of the output itself). The following definitions standardise the relative placements of the various inputs and outputs.

\begin{defn}
\label{def1a}
A machine learning model with a quantum circuit acting as a resource is a map of the form: 
\begin{center}\begin{quantikz}
            \lstick{$\ket{0}$} & \qwbundle{n_0}  & \qw & \qw  & \gate[4]{ P } \hphantom{wide}  & \qwbundle{n_y}  \rstick{$\ket{y}$} \\  
            \lstick{$\ket{0}$} & \qwbundle{n_q}  & \gate[2]{ U } \hphantom{wide} & \meter{} & &   & \\      
            \lstick{$\ket{x}$} & \qwbundle{n_x}  &   & \meter{}   & & \\ 
            \lstick{$\ket{+}$} & \qwbundle{n_+}   & \qw & \meter{} & & \qw  \rstick{$\ket{\tilde{y}}$}   \\
\end{quantikz} 
\end{center}

equivalently,

\begin{center}
\begin{quantikz}
            \lstick{$\ket{0}$} & \qwbundle{n_0}  & \qw & \gate[4]{ P } \hphantom{wide}   & \qwbundle{n_y}  & \meter{} \rstick{$\ket{y}$} \\  
            \lstick{$\ket{0}$} & \qwbundle{n_q}  & \gate[2]{ U } \hphantom{wide} & & &    \\      
            \lstick{$\ket{x}$} & \qwbundle{n_x}  &   & & & \\  
            \lstick{$\ket{+}$} & \qwbundle{n_+}   & \qw &  & \qw & \meter{}  \rstick{$\ket{\tilde{y}}$}   
\end{quantikz} 
\end{center}
where $P$ is any permutation.
\end{defn}

\begin{defn}
\label{def1b}
A machine learning model with a quantum state acting as a resource is a map of the form:
\begin{center}\begin{quantikz}
            \lstick{$\ket{0}$} & \qwbundle{n_0}  & \qw & \qw  & \gate[4]{ P } \hphantom{wide}  & \qwbundle{n_y}  \rstick{$\ket{y}$} \\  
            \lstick{$\ket{0}$} & \qwbundle{n_q}  & \gate[1]{ U }  & \meter{} & &   & \\   
            \lstick{$\ket{+}$} & \qwbundle{n_+}   & \qw & \meter{} & & & \\  
            \lstick{$\ket{x}$} & \qwbundle{n_x}  & \qw  & \qw   & & \qw  \rstick{$\ket{\tilde{y}}$}   
\end{quantikz} 
\end{center}

equivalently,

\begin{center}
\begin{quantikz}
            \lstick{$\ket{0}$} & \qwbundle{n_0}  & \qw & \gate[4]{ P } \hphantom{wide}   & \qwbundle{n_y}  & \meter{} \rstick{$\ket{y}$} \\  
            \lstick{$\ket{0}$} & \qwbundle{n_q}  & \gate[1]{ U }  & & &    \\    
            \lstick{$\ket{+}$} & \qwbundle{n_+}   & \qw &  & &  \\
            \lstick{$\ket{x}$} & \qwbundle{n_x}  & \qw  & & \qw  & \meter{}  \rstick{$\ket{\tilde{y}}$}   
\end{quantikz} 
\end{center}
where $P$ is any permutation.
\end{defn}

\begin{remark}
    The equivalence in these definition is a direct implication of Lemma~\ref{gen-defer}.
\end{remark}

\begin{remark}
    These structures may be thought of capturing every machine learning model \textbf{up to and including} the stated size (in terms of $n_0$, $n_+$, $n_q$, $n_x$ and $n_y$) as certain (qu)bits can just be left unused.
\end{remark}

Why is it that these model structures wholly capture all the ways in which an untrained quantum circuit can act as a resource in a machine learning model? First, if the input is classically pre-processed ahead of passing through the quantum circuit, then this can be treated as a permutation preceding $U$ that is then absorbed thereinto. Finally, Refs.~\cite{fujii2017quantumreservoir, Fujii2021, wilson2020quantum} all take observable measurements of the untrained quantum circuit as the input to the classical machine learning model, however this is also captured as the circuit $U$ could first copy the \textit{classical} input $\ket{x}$, and then operate on the various copies in parallel, with the observable estimation also being absorbed into $U$.

\section{The range and cost of achievable computational effects is independent of an untrained quantum circuit used as a resource}
\label{sec:iii}

Henceforth for brevity we simply refer to the titular claim of this section as ``the independence theorem''.

\subsection{The cost of preparing the family of target distributions}
\label{subsec:cost-target}

To prove the independence theorem, it is first incumbent upon us to define what exactly is meant by the cost incurred by each model. Let $n = n_0 + n_+ + n_q + n_x$ be the total number of (qu)bits, and further let $\mathcal{P}_n$ be the elements of the symmetric group $S_n$ represented as matrices. Further let $c(P)$ be the \textit{cost} of preparing some $P \in \mathcal{P}_n$ (for given $n$). $c(\cdot)$ can be defined as any cost that depends only on $P$, and for instance could be,
\begin{enumerate}[(i)]
    \item the minimum number of gates from some gateset (that is universal for permutations) that are required to implement $P$ as a circuit;
    \item a binary value $0$ if a circuit to prepare $P$ with fewer gates than some threshold value exists, and $1$ otherwise.
\end{enumerate}
The second of these has the operational interpretation of counting the number of problems that can be computed within some model capacity. Hence this directly corresponds to a measure on expressivity, and so the forthcoming independence theorem captures as a special case the question of whether certain resource states / circuits enhance the expressivity on average.

We further need to define a notion of equivalence of permutations, which is model-dependent. 
\begin{defn}
\label{def:equivalence}
Let $\mathcal{Z}$ be a set of input states of size $n$ (which could be a singleton), then let $\mathcal{P}_{\mathcal{Z}, n_y}$ be a set of \textbf{permutation equivalence classes}, where 
\begin{widetext}
    \begin{align}
        \forall P, Q  \in \mathcal{P}_n, \exists \mathrm{P} \in \mathcal{P}_{\mathcal{Z}, n_y},  & \, P, Q \in \mathrm{P} \nonumber \\
        & \iff
        \forall \ket{z} \in \mathcal{Z}, \forall y \in \{0,1\}^{n_y}, \, \sum_{\tilde{y} \in \{0,1\}^{n - n_y}}
        | \bra{z} P \ket{y \tilde{y}}|^2 = \sum_{\tilde{y} \in \{0,1\}^{n - n_y}}
        | \bra{z} Q \ket{y \tilde{y}}|^2
    \end{align}
    \end{widetext}
    That is, for every input, the output statistics are the same whichever of $P$ or $Q$ is applied if and only if they are in the same equivalence class. 
\end{defn}
It is trivial to show that these do indeed constitute equivalence classes. The permutation equivalence classes should be interpreted as the computational effects that can be achieved with any given model, and as such only one permutation from each equivalence class need be prepared. It is most efficient to choose the minimal possibility (in terms of the defined cost), and this gives a natural definition of the cost of an equivalence class, and hence the overall cost.

\begin{defn}
\label{def:cost}
    Let $M$ be a positive integer such that for some $\mathcal{Z}$ and $n_y$, $\mathcal{P}_{\mathcal{Z}, n_y} = \{\mathrm{P}_1, \mathrm{P}_2, \dots, \mathrm{P}_M \}$, then define the cost, $\tilde{c}$ for each equivalence class as for every $i \in \{ 1, M\}$,
    \begin{equation}
        \tilde{c}(\mathrm{P}_i) = \min_{\substack{P \in \mathrm{P}_i}} c(P)
    \end{equation}
    and further define the overall cost as
    \begin{equation}
    \label{eq:cost-gen-modela}
        C(\mathrm{P}_1, \mathrm{P}_2, \dots, \mathrm{P}_{M}) = \frac{1}{M} \sum_{i=1}^{M} \tilde{c}(\mathrm{P}_i)
    \end{equation}
\end{defn}
Note that below we will usually drop the parameterisation and refer to the overall cost simply as $C$. Note also that as the resource circuit $U$ is fixed (rather than varying with the problem set, rather in a sense it defines the problem set) preparation thereof is excluded from the cost function.
\begin{remark}
    The independence theorem holds for any $C(\mathrm{P}_1, \mathrm{P}_2, \dots, \mathrm{P}_{M})$ that is symmetric in its arguments.
\end{remark}

\subsection{Preliminaries}
\label{sec:prelims}

We now give some further preliminary results needed for the independence theorem theorem. In the first instance (before later generalisation) we consider the case where $n_x = 0$, and so the models of Definitions~\ref{def1a} and \ref{def1b} are equivalent. This has the operational interpretation of a generative model, where a sample obtained by measuring a quantum state is (part of) in the latent space. As such, we define $\ket{\psi} = U \ket{0^{n_q}}$, and refer to this as the ``resource state''.

As the permutation $P$ acts on the resource state along with the other inputs, in the following analysis we define the input to $P$ as 
\begin{equation}
    \ket{\tilde{\psi}} = \ket{0^{n_0}} (U \ket{0^{n_q}}) \ket{+^{n_+}} \ket{x} = \ket{0^{n_0}} \ket{\psi}\ket{+^{n_+}} \ket{x}
\end{equation}
where the second equality follows from the prior definition of $\ket{\psi}$.

In the following analysis we frequently interchange between quantum states and the corresponding complex unit vectors, and hence it is convenient to conflate these to some extent. In particular, the subscript $i$ in $\ket{\psi}_i$ is to be read as referring to the $i$th element of this vector of the state expressed in the computational basis. In particular, as we are considering the case where $n_x = 0$, the input to $P$ is the state / vector:
\begin{widetext}
\begin{equation}
    \label{eqn:input-state}
        \ket{\tilde{\psi}}=\frac{1}{\sqrt{2^{n_+}}} \Bigg[ \,\,\,\, \underbrace{\ket{\psi}_1, \ket{\psi}_1, \dots \ket{\psi}_1}_{2^{n_+} \text{ times}}, \,\,\,\, \underbrace{\ket{\psi}_2, \ket{\psi}_2, \dots \ket{\psi}_2}_{2^{n_+} \text{ times}}, \,\,\,\, \dots \underbrace{\ket{\psi}_{2^{n_q}}, \,\,\,\, \ket{\psi}_{2^{n_q}}, \dots \ket{\psi}_{2^{n_q}}}_{2^{n_+} \text{ times}}, \underbrace{0, 0, \dots, 0}_{N-2^{n_+ + n_q} \text{ times}} \Bigg]^T
    \end{equation}
    \end{widetext}
where $N = 2^{n}$, i.e., $\tilde{\psi}$ is an $N$-element vector.

It is further helpful to define:

\begin{defn}
    Let $\{ x_i \}$ be a set of complex numbers. Call $\sqrt{ \sum_i |x_i|^2 }$ the $\ell_2$ sum of $x$ 
\end{defn}
Usually in the following, the set will be the some or all of the elements of some quantum state expressed as a computational basis state vector. Indeed, in the case where the set is all elements of a state, then this is simply the well-known $\ell_2$-norm. We now move on to some important properties of quantum states, that we will later see must be met in the independence theorem.

\begin{defn}
\label{def:distinct}
Call a $n$-qubit quantum state distinct if no two elements of the corresponding $2^n$-element computational basis vector have the same magnitude.
\end{defn}

\begin{defn}
\label{def:strong-distinct}  
    Let $\Phi$ be a multiset whose elements are the elements of $D$ copies of some $n$-qubit state, $\ket{\psi}$ expressed as a computational basis state vector. Call $\ket{\psi}$ $D$-distinct if the $\ell_2$ sum of two sub-multisets of $\Phi$ being equal implies that the sub-multisets are themselves equal.
\end{defn}

Distinctness and $D$-distinctness are important to avoid probabilities ``accidentally'' summing to equal amounts, which would break the structure necessary for the independence theorem in its most general form. 

\begin{example}
    Let $\ket{\psi} = [ \sqrt{1 / 3}, \sqrt{2 / 3}]^T$, then $\ket{\psi}$ is 1-distinct, but not 2-distinct. The former case is immediate, and in the latter case $\Phi = \{ \sqrt{1 / 3}, \sqrt{2 / 3}, \sqrt{1 / 3}, \sqrt{2 / 3} \}$, so one has unequal sub-multisets $\{\sqrt{2 / 3} \} $ and $\{ \sqrt{1 / 3}, \sqrt{1 / 3} \}$ whose elements $\ell_2$ sum to the same value.
\end{example}

We next define two sets of block matrices, $\mathcal{V}$ and $\mathcal{W}$.

\begin{defn}
$\mathcal{V}$ is the set of $N \times N$ matrices, where $V \in \mathcal{V}$ if and only if:
\begin{equation}
\label{eqn:q-eqn}
    V = \begin{bmatrix} \tilde{V}_1 & 0 & \dots & 0 & 0 \\ 0 & \tilde{V}_2 & \dots & 0 & 0  \\ \vdots & \vdots & \ddots & 0 & 0 \\ 0 & 0 & 0 & \tilde{V}_{2^{n_q}} & 0 \\ 0 & 0 & 0 & 0 & \tilde{V}_{2^{n_q}+1} \end{bmatrix} 
\end{equation}
where $\tilde{V}_1 \dots \tilde{V}_{2^{n_q}}$ are any $2^{n_+} \times 2^{n_+}$ permutation matrices, and $\tilde{V}_{2^{n_q}+1}$ is any $(N-2^{n_+ + n_q}) \times (N-2^{n_+ + n_q})$ permutation matrix. 
\end{defn}
\begin{remark}
\label{rem:group1}
    $\mathcal{V}$ forms a group under matrix multiplication, and in particular is closed under multiplication and inversion.
\end{remark}
\begin{defn}
$\mathcal{W}$ is the set of $N \times N$ block diagonal matrices where $W \in \mathcal{W}$ if and only if:
\begin{equation}
\label{eqn:s-eqn}
    W = \begin{bmatrix} \tilde{W}_1 & 0 & \dots & 0  \\ 0 & \tilde{W}_2 & \dots & 0   \\ \vdots & \vdots & \ddots & 0  \\ 0 & 0 & 0 & \tilde{W}_{2^{n_y}}  \end{bmatrix} 
\end{equation}
where $\tilde{W}_1 \dots \tilde{W}_{2^{n_y}}$ are any $2^{n-n_y} \times 2^{n-n_y}$ permutation matrices.
\end{defn}
\begin{remark}
\label{rem:group2}
    $\mathcal{W}$ forms a group under matrix multiplication, and in particular is closed under multiplication and inversion.
\end{remark}

For states / vectors to which matrices in $\mathcal{W}$ are applied, it is helpful to think of these vectors as themselves being (implicitly) decomposed into blocks of size $2^{n - n_y}$, such that only elements in the same block are permuted. We now use these classes of matrices to define a second notion of equivalence:

\begin{defn}
\label{def:multipl-equiv}
    Let $P$ and $Q$ be $N \times N$ permutation matrices. We say $P$ and $Q$ belong to the same \textbf{multiplicative equivalence class} if $P$ can be expressed as $P=WQV$ for some $W \in \mathcal{W}$ and $V \in \mathcal{V}$.
\end{defn}

\begin{remark}
\label{rem:multi-equiv-membership}
    Specifying $n_0$, $n_+$, $n_q$, and $n_y$ wholly defines the multiplicative equivalence classes.
\end{remark}

The essence of the no independence theorem is to show that the multiplicative and permutation equivalence classes coincide in the setting of interest, and hence that the total cost is independent of the resource state, $\ket{\psi}$. For this, most of the heavy-lifting is done in the following lemma.

\begin{lemma}
    \label{lem4}
Let $\ket{\psi}$ be an $n_q$-qubit quantum state, and $\mathcal{Z} = \{ \ket{\tilde{\psi}} \}$ for $\ket{\tilde{\psi}}$ defined in (\ref{eqn:input-state}), such that permutation equivalence classes are defined as in Definition~\ref{def:equivalence}. For two $N \times N$ permutation matrices $P$ and $Q$:
\begin{enumerate}[(i)]
    \item If $P$ and $Q$ are in the same multiplicative equivalence class, then they are in the same permutation equivalence class.
    \item If $\ket{\psi}$ is distinct and $2^{n_+}$-distinct, if $P$ and $Q$ are in the same permutation equivalence class, then they are in the same multiplicative equivalence class.
\end{enumerate}
\end{lemma}
\begin{remark}
    Therefore for any $\ket{\psi}$ which is distinct and $2^{n_+}$-distinct, it follows that $P$ and $Q$ are in the same permutation equivalence class if and only if they are in the same multiplicative equivalence class. 
\end{remark}
\begin{proof}
    Using the decomposition $P = WQV$ for some $W \in \mathcal{W}$ and $V \in \mathcal{V}$ (by multiplicative equivalence) and considering $P \ket{\tilde{\psi}}$, we can see $V\ket{\tilde{\psi}} = \ket{\tilde{\psi}}$ as the operation of $V$ is simply to permute elements that are, by construction, equal. We also have $V\ket{\tilde{\psi}} = \ket{\tilde{\psi}} \implies QV \ket{\tilde{\psi}} = Q \ket{\tilde{\psi}} \implies P \ket{\tilde{\psi}} = WQV \ket{\tilde{\psi}} = WQ \ket{\tilde{\psi}}$. Owing to the fact that only the first $n_y$ qubits are measured to form the sample, any permutation of the elements $\{c \times 2^{n-n_y} + 1, (c+1) \times 2^{n-n_y}\}$ of the prepared state for any $c \in \{0, 2^{n_y}-1 \}$ does not change the measurement statistics. So it follows that for any $W \in \mathcal{W}$ and permutation matrices $S$ and $T$, $WT$ and $S$ are in the same permutation equivalence class if and only if $T$ and $S$ are in the same permutation equivalence class. Together with the fact, established above, $P \ket{\tilde{\psi}} = WQ \ket{\tilde{\psi}}$, this is enough to conclude claim (i).

    Turning to claim (ii), define $T$, a permutation, such that $P = TQ$. Without loss of generality let $T$ be expressed as a product of alternating terms $\{W_i\}$ and $\{T_i\}$ where $\forall i, \, W_i \in \mathcal{W}$, i.e., it just swaps elements in each $2^{n - n_y}$ block, and $\forall i , \, T_i$ is a expressed as a product of transpositions which swap elements in different $2^{n - n_y}$ blocks. That is, $T = \prod_i W_i T_i$. Every $T_i W_{j}$ can be written as equal to $W_{j} \tilde{T}_i$ where $\tilde{T}_i$ itself is a product of transpositions which swap elements in different $2^{n - n_y}$ blocks. Doing this multiple times, one gets $T = W \tilde{T}$ where $W = \prod_i W_i$ and is in $\mathcal{W}$ by closure under multiplication, and $\tilde{T} = \prod_i \tilde{T}_i$ is a product of transpositions which swap elements in different $2^{n - n_y}$ blocks.

    Returning to the definition $P = TQ = W \tilde{T} Q$, it is the case that $P$ and and $\tilde{T} Q$ are in same permutation equivalence class, by the same rationale given in the proof of claim (i), above. However, recall that $\tilde{T}$ is a product of transpositions of elements that swap elements in different $2^{n - n_y}$ blocks, by $2^{n_+}$-distinctness, $\tilde{T} Q$ and $P$ can only prepare the states with the same measurement statistics (i.e., be permutation equivalent) if the overall effect of $\tilde{T}$ is to swap elements of the same value between blocks. This in turn means that $\tilde{T}$ can be written $\tilde{T} = \tilde{W} \tilde{V}$ where $\tilde{W} \in \mathcal{W}$ and $\tilde{V}$ is a product of transpositions swapping equal-valued elements. 

    Thus we have $P = W \tilde{T} Q = W \tilde{W} \tilde{V} Q$, and the final step is to notice that writing $\tilde{V} Q = Q V$, then $V \in \mathcal{V}$ as it swaps elements of the same value in the original ordering of the input, and by distinctness these must therefore be (scaled) repeats of the same element of $\ket{\psi}$ once the tensor product is taken. Thus $P =  W \tilde{W} S V$ where $V \in \mathcal{V}$ and $W, \tilde{W} \in \mathcal{W} \implies W \tilde{W} \in \mathcal{W}$ (by closure under matrix multiplication), completing the proof of claim (ii).
\end{proof}

An implication of this lemma is that the number of multiplicative equivalence classes upper-bounds the number of permutation equivalence classes, and so the multiplicative equivalence classes can be thought of as the maximum set of computational effects that can be achieved.

\subsection{Main result: the independence theorem}
\label{sec:main}

We are now ready to prove the independence theorem, starting with a formal definition thereof.

\begin{defn}
    Say that the range and cost of achievable computational effects is independent of an untrained quantum circuit used as a resource if: for a quantum circuit, $U$, used as a resource in a machine learning model in the sense of Definition~\ref{def1a} or \ref{def1b}, any other quantum circuit $\tilde{U}$ that can achieve the same set of computational effects, does so with the same overall cost, as defined in Def~\ref{def:cost}. 
\end{defn}

\begin{thm}
\label{nfl1}
For a machine learning model with $n_x = 0$ as defined in either Definition~\ref{def1a} or \ref{def1b} (which are equivalent in this case), if $U \ket{0}$ is distinct and $2^{n_+}$-distinct, then the range and cost of achievable computational effects is independent of an untrained quantum circuit used as a resource.
\end{thm}
\begin{proof}
    By Lemma~\ref{lem4}, as $\ket{\psi} = U \ket{0}$ is distinct and $2^{n_+}$-distinct, it follows that multiplicative equivalence classes are the same as the permutation equivalence classes, and so one permutation from each multiplicative equivalence class is needed in the the overall cost defined in Definition~\ref{def:cost}. However, by Remark~\ref{rem:multi-equiv-membership}, the multiplicative equivalence classes depend only on the parameters $n_0$, $n_+$, $n_q$ and $n_y$ and are independent of the resource state. As the maximum number of computational effects that can be achieved is one for each multiplicative equivalence class, the same conclusion holds for every alternative resource state that achieves the same set of computational effects, and so the range and cost of achievable computational effects is independent of an untrained quantum circuit used as a resource.
\end{proof}

We now turn to the generalisations of the independence theorem for cases where the input is non-empty (i.e., $n_x > 0$), and for this we again use the fact that if the set of computational effects (i.e., the set of permutation equivalence classes) can be defined such that it is independent of the input state, then we can conclude that the range and cost of achievable computational effects is independent of an untrained quantum circuit used as a resource.

\begin{corollary}
For a machine learning model as defined in Definition~\ref{def1a}, with any $n_x$, if $U (\ket{0}\ket{x})$ is distinct and $2^{n_+}$-distinct for at least one $\ket{x}$, then the range and cost of achievable computational effects is independent of an untrained quantum circuit used as a resource.
\end{corollary}
\begin{proof}
Let the input $\ket{x^*}$ be such that $\ket{\psi} = U (\ket{0}\ket{x^*})$ is distinct and $2^{n_+}$-distinct. If the input were instead fixed, then  Theorem~\ref{nfl1} would apply. This means that in order to obtain the full range of computational effects as $\ket{x}$ is allowed to vary as an input, it remains the case that one member of each multiplicative equivalence class must be chosen, and so the independence theorem still applies in this case.
\end{proof}

\begin{corollary}
For a machine learning model as defined in Definition~\ref{def1b}, with any $n_x$, if $U\ket{0}$ is distinct and $2^{n_+}$-distinct, then the range and cost of achievable computational effects is independent of an untrained quantum circuit used as a resource.
\end{corollary}
\begin{proof}
Consider a smaller model structure with the same $U$, $n_0$, $n_+$ and $n_y$, but with $n_x = 0$ and define $\tilde{n} = n_0 + n_+ + n_q$. Let $\tilde{\mathcal{Z}} = \{ \ket{\tilde{\psi}} \}$ where $\ket{\psi} = U \ket{0^{n_q}}$ and $\ket{\tilde{\psi}}$ is defined as in (\ref{eqn:input-state}). If we now consider the full model structure with input $\ket{x}$, such that $\mathcal{Z} = \{ \ket{0^{n_0}} \ket{\psi} \ket{+^{n_+}} \ket{x} \}_{ x \in \{0,1\}^{n_x}}$, then we have the set of permutation equivalence classes
\begin{equation}
\label{erqn:8}
    \mathcal{P}_{\mathcal{Z}, n} = \left\{ \sum_x \mathrm{P}_x \otimes \ket{x} \bra{x} \right\}_{\mathrm{P}_x \in \mathcal{P}_{\tilde{\mathcal{Z}}, \tilde{n}}}
\end{equation}
By distinctness and $2^{n_+}$-distinctness, Theorem~\ref{nfl1} applies to the smaller model, and so the equivalence classes of $\mathcal{P}_{\tilde{\mathcal{Z}}, \tilde{n}}$ coincide with the multiplicative equivalence classes, and are independent of the resource state. By (\ref{erqn:8}), this property is inherited by the equivalence classes of $\mathcal{P}_{\mathcal{Z}, n}$ and so the range and cost of achievable computational effects is independent of an untrained quantum circuit used as a resource.
\end{proof}

\subsection{Equating ``untrained'' with ``Haar random''}

The independence theorems rely on the fact that the resource state is distinct and $2^{n_+}$ distinct, and it is worth pausing to consider whether this is a realistic and appropriate restriction. We will now see that in one sense at least, it is -- if we consider the circuit that prepares the resource state / acts on the input as not merely untrained, but randomly chosen. A suitable and standard notion of randomness of unitary operators is sampling according to the Haar measure.
\begin{defn}[Ref.~\cite{mezzadri2007generate}]
\label{def:haar-U}
    A Haar random unitary, $U$, can be constructed by first filling an appropriately sized matrix, $A$, with elements of the form $a_{i,j} + b_{i,j}\mathrm{i}$ where $a_{i,j}$ and $b_{i,j}$ are sampled independently, and identically distributed (iid) from a unit Gaussian\footnote{That is, sampling from the Ginibre ensemble \cite[Section 3]{mezzadri2007generate}} (and $\mathrm{i}= \sqrt{-1}$); then performing a QR decomposition on $A$, that is such that $A = QR$ where $Q$ is unitary and $R$ is upper-triangular; and finally setting $U = Q \Lambda$ where $\Lambda = \textnormal{diag} (R_{ii}/|R_{ii}|)$.
\end{defn}

\begin{defn}
\label{def:haar-state}
    An $n$-qubit Haar random state may be sampled by applying a $2^n \times 2^n$ Haar random unitary to any fixed $n$-qubit input state.
\end{defn}
Visually, a Haar random state can be thought of as uniformly sampling over the surface of the a unit-hypersphere in the Hilbert space of the relevant dimension.

\begin{lemma}
\label{lem1}
Let $\ket{\psi}$ be a Haar random state, then $\ket{\psi}$ is distinct and $D$-distinct for any finite constant $D$, almost surely. 
\end{lemma}
\begin{proof}
The set of points on the surface of the unit hypersphere where multiple elements are equal in magnitude is Haar measure-0, demonstrating distinctness. 

In order for $D$-distinctness to be violated, only finitely many $\ell_2$ sums of subsets of elements need be compared. Thus (by the union bound) it is enough to show that each of these violations is Haar measure-0. This holds, as the points on the unit hypersphere surface where some particular collection of elements $\ell_2$ sums to that of a different collection of elements is Haar measure-0.
\end{proof}

\begin{corollary}
    If $U$ in (i) Definition~\ref{def1a} or (ii) Definition~\ref{def1b} is obtained by a Haar random sample, then there is the range and cost of achievable computational effects is independent of an untrained quantum circuit used as a resource, almost surely.
\end{corollary}

Having proved these mathematical facts it is worth addressing the fair potential criticism that Haar randomness is rather a strong assumption. To partially alleviate such concerns, it is therefore worth pointing out that sampling uniformly from the Haar measure is sufficient (almost surely) but certainly far from necessary for the result to hold, and a future direction of research could be to identify other circuit classes that enjoy the properties of distinctness and $2^{n_+}$ distinctness.

\section{No untrained quantum circuit is very useful on average}
\label{sec:supp-analysis}

There are two key idealisations in the independence theorem. Firstly, that the state is distinct and $2^{n_+}$-distinct, and secondly that every possible computational effect is treated as equally important. Taking these two in turn, we may first ask whether distinctness is really required, and a simple counter-example (to the proposition to omit the requirement of distinctness) can be constructed to show that, yes, it is.

\begin{example}
Let $n_0 = n_x = n_+ = 0$, $n_q = 3$ and $n_y$ = 1. For simplicity, we consider the output to be the probability that $\ket{y}$ yields the value $1$ when measured, and label this $p_y$. We treat the gateset as the set of all computational basis state transpositions (which is universal for computational basis state permutations). We now consider two alternative resource states, $\ket{\psi}$ and $\ket{\omega}$ that generate the same set of outputs, but with different average costs. In particular, let:
\begin{equation*}
    \ket{\psi} = \begin{bmatrix} \sqrt{0.3} \\ \sqrt{0.3}  \\ \sqrt{0.3} \\ \sqrt{0.1} \\ 0 \\ 0 \\ 0 \\ 0 \end{bmatrix} \,\,\,\,\,\,\,\,\,\,\,\,\,\,\,\,\,\,\,\,\,\,\,\, \ket{\omega} = \begin{bmatrix} \sqrt{0.3} \\ \sqrt{0.3} \\ 0 \\ 0  \\ \sqrt{0.3} \\ \sqrt{0.1} \\ 0 \\ 0  \end{bmatrix}
\end{equation*}
Consider (example) minimal circuits to prepare a state sampling $y=1$ with each $p_y$, for initial state $\ket{\psi}$
    
    \begin{center}
    \begin{tabular}{c c c c c }
        $p_y$ & & $\textnormal{Transpositions} $ & & $\tilde{c} $ \\
        \hline
        \hline
         $1$ & & $-$ & & $0$  \\
         $0.9$ & & $(4,5)$ & & $1$ \\
         $0.7$ & & $(3,5)$ & & $1$ \\
         $0.6$ & & $(3,5) , (4,6)$ & & $2$   \\
         $0.4$& & $(2,5), (3,6)$ & & $2$  \\
         $0.3$ & & $(2,5), (3,6), (4,7)$ & & $3$ \\
         $0.1$ & & $(1,5), (2,6), (3,7)$ & & $3$  \\
         $0$ & & $(1,5), (2,6), (3,7), (4,8)$ & & $4$ 
    \end{tabular}
    \end{center}
For initial state $\ket{\omega}$

\begin{center}
    \begin{tabular}{c c c c c}
        $p_y$ & & $\textnormal{Transpositions}$ & & $\tilde{c} $ \\
        \hline
        \hline
         $1$  & & $(3,5), (4,6)$ & & $2$ \\
         $0.9$  & & $(3,5)$ & & $1$ \\
         $0.7$  & & $(3,6)$ & & $1$ \\
         $0.6$  & & $-$ & & $0$ \\
         $0.4$  & & $(2,6)$ & & $1$ \\
         $0.3$ & & $(2,7)$ & & $1$ \\
         $0.1$  & & $(1,6), (2,7)$ & & $2$ \\
         $0$ & & $(1,7), (2,8)$ & & $2$
    \end{tabular}
    \end{center}
%
Therefore $C(\ket{\psi}) = 2$ and $C(\ket{\omega}) = \frac{10}{8} = 1.25$
\end{example}

Turning to the second idealisation, theoretical computer science analysis typically allows outputs to be ``$\epsilon$-close'' to that which is desired. Ostensibly, the implication of this may appear to be a coarsening of the output, where certain computational effects (i.e., certain permutations) are treated as the same. However, the non-transitivity of $\epsilon$-closeness (for example, the numbers $0$ and $1$ are both $0.75$-close to the number $0.5$, but not $0.75$-close to each other) does not allow for a straightforward merging of equivalence classes as one may hope.

In particular, whilst the equivalence class structure remains, in the $\epsilon$-closeness setting it may be the case that a set of computational effects \textit{not} exhausting the equivalence classes (i.e., not with exactly one member per equivalence class) is sufficient. Furthermore, it is not necessarily the case that for every model that achieves the desired function, the same equivalence classes may always be omitted from consideration. Thus the $\epsilon$-closeness setting motivates a suitably relaxed notion of cost.

\begin{defn}
    For given $\mathcal{Z}$ and $n_y$, let $\mathcal{P}'$ be a subset of equivalence classes in $\mathcal{P}_{\mathcal{Z}, n_y}$ such that every computational effect (up to some desired accuracy) is achieved by at least one element of $\mathcal{P}'$. Further let $\tilde{C}$ be the minimal $C(\mathcal{P}')$ over all such subsets $\mathcal{P}'$.
\end{defn}

In order to isolate any potential enhancements offered by the quantum resource, in the following analysis we set $n_+ = 0$ (if one likes, one can take the quantum resource sample as subsuming any uniform ``classical'' randomness, hence there is no real loss of generality in so doing). We furthermore consider two cases: (i) there is an input (which may pass through a fixed quantum circuit acting as a nonlinear map), and we are in particular concerned with the case where any injected resource is up to about the same size as the input, hence we set $n_q \in \mathcal{O} (n_x)$; (ii) there is no input, e.g., the circuit is a generative model mapping a latent distribution of size $n_q$ to a target distribution of size $n_y = \Theta (n_q)$. For these two cases, we define the relaxed costs $\tilde{C}_\text{input}$ and $\tilde{C}_\text{no input}$, respectively.

We further define $n_\text{in} = n_q + n_x$, and observe that $n_\text{in} = \Theta (n_x)$ in the first case and $n_\text{in} = n_q$ in the second; we also assume that $n_y \leq n_\text{in}$, which can be justified by a simple information theoretic argument -- that any greater output size would be artificial in some sense, and certainly compressible into at most $n_\text{in}$ bits. Finally, we make the standard assumption that the entire workspace is polynomial in the input, and hence $n \in \text{Poly} (n_\text{in})$ (and in the first case where there is an input, $ n \in \text{Poly} (n_x)$). The following subsections take the two cases in turn.

\subsection{A quantum circuit as an input map and / or a quantum state as an injected resource}

In this setting the quantum circuit acts as resource, as in Definition~\ref{def1a}, and by the above assumption $n_+ = 0$. 

\begin{thm}
\label{thm:extra1}
    $\tilde{C}_\textnormal{input} \in \tilde{\Omega} (\exp(n_{x}))$.
\end{thm}
\begin{proof}
Regardless of which outputs are treated as $\epsilon$-close, for any meaningful computation, there \textit{is} always an output, and hence $n_y \geq 1$, and so there are at least two possible outputs. As each of the $2^{n_{x}}$ inputs may map to any output, there are at least $2^{2^{^{n_{x}}}}$ computational effects, each of which must be prepared by a different circuit. Consider the minimal resources needed to prepare half of the different circuits, (i.e., at least $ \frac{1}{2} 2^{2^{^{n_{x}}}} $ different circuits) using gates from some fixed gateset of size $n_g$, each of which acts on at most $n_d$ (qu)bits ($n_d$ would typically be 2 or 3 in practice). Therefore there are at most $(n_g n^{n_d})^{n_c}$ functionally distinct circuits of exactly $n_c$ gates, and at most  
\begin{equation}
    \sum_{i = 1}^{n_c} (n_g n^{n_d})^i \leq (n_g n^{n_d})^{n_c + 1}
\end{equation}
functionally distinct circuits of at most $n_c$ gates. 

We now wish to lower-bound $n_c$ needed to prepare $\frac{1}{2} 2^{2^{^{n_{x}}}} $ different circuits:
\begin{align}
    (n_g n^{n_d})^{n_c + 1} \geq & \frac{1}{2} 2^{2^{n_{x}}} \nonumber \\
    \implies (n_c + 1) \log_2 (n_g n^{n_d}) \geq & 2^{n_{x}} -1  \nonumber \\
    \implies n_c \geq & \frac{2^{n_{x}} - 1}{\log_2 (n_g n^{n_d})} -1 \nonumber \\
\end{align}
and so when $n_d$ and $n_g$ are constants and $n \in \text{Poly}(n_x)$, it is the case that to prepare $\frac{1}{2} 2^{2^{n_{x}}}$ computational effects, $n_c \in \tilde{\Omega} (\exp(n_{x}))$. Thus, at least half of the computational effects must take more gates than this, and so $\tilde{C}_\text{input} \in \tilde{\Omega} (\exp(n_{x}))$. 
\end{proof}

On the other hand,
\begin{thm}
\label{thm:extra2}
    Circuits exist such that  $\tilde{C}_{\textnormal{input}} \in \tilde{\mathcal{O}} (\exp (n_{x}))$.
\end{thm}
\begin{proof}
Using the standard reversible circuit structure, 
\begin{center}\begin{quantikz}
            \lstick{$\ket{0}$} & \qwbundle{n_y}  & \gate[2]{ P } \hphantom{wide}  & \qw  \rstick{$\ket{y(\text{input})}$} \\ 
            \lstick{$\ket{\text{input}}$} & \qwbundle{n_{\text{in}}}  & & \qw  \rstick{$\ket{\text{input}}$} 
\end{quantikz}
\end{center}
every computational effect can be prepared as a permutation on $n_{\text{in}} + n_y = \Theta (n_x)$ elements. By a standard result, the largest circuit can be expressed as a product of some $2^{\Theta (n_{x})} - 1$ transpositions at most. Each transposition can be implemented in $\mathcal{O}(n)$ gates \cite{herbert2023almostoptimal}, and so the total depth of any circuit is at most $\mathcal{O}(n \exp (n_{x})) \in \tilde{\mathcal{O}}(\exp (n_{x}))$ gates.
\end{proof}

Theorem~\ref{thm:extra1} says that the range of computational effects (even with the idealisations needed for the independence theorem dropped) is so large that however the input is represented (perhaps passed through an untrained quantum circuit to act as a nonlinear map), the average circuit size needed is (nearly) exponentially large \textit{at best}, whilst Theorem~\ref{thm:extra2} says that even using a naive way to implement a function, the circuit size need only be a little more than exponential \textit{at worst}. So this backs up the independence theorem in a more practically-relevant setting.

\subsection{A (measured) quantum state as the latent distribution in a generative model}

In this setting, the error must explicitly be used to count the number of distinct possible outputs (as there is no input as such to lower-bound this quantity as in the previous case). As is standard practice, total variation distance can be used to measure $\epsilon$-closeness of distributions (e.g., Ref.~\cite[Def. 11]{Aaronson2014Equivalence}).

\begin{lemma}
\label{lem:lem}
There are at least some $\frac{2^{0.189 \times 2^{n_y}}}{2^{n_y} + 1}$ distributions over $\{0,1\}^{n_y}$ that are pairwise $\epsilon$-distinct in total variation distance, for any $\epsilon \leq \frac{1}{4}$.   

\end{lemma}
\begin{proof}
    We can adapt standard analysis from coding theory (see e.g., Refs.~\cite{Roth2006, cornell2019gilbertvarshamov}). First let $N_y = 2^{n_y}$ and consider just those distributions that are uniform over some $\frac{N_y}{2}$ bitstrings. There are ${N_y \choose N_y/2} \geq \frac{2^{N_y}}{N_y+1}$ such distributions.  

    To satisfy the $\epsilon$-distinctness condition, consider each distribution as an indicator bitstring, where $N_y$ indicator bits signify inclusion of each possible $n_y$-bitstring. For $\frac{1}{4}$-distinctness it is enough that the pairwise Hamming distance of the indicator bitstrings is at least $\frac{N_y}{4}$. A Hamming ball of radius $\frac{N_y}{4}$ contains some $\sum_{j=0}^{\lfloor N_y / 4 \rfloor} {N_y \choose j} \leq 2^{N_y \mathcal{H}_2 \left( \frac{1}{4} \right)}$ binary strings, where $\mathcal{H}_2$ is the binary entropy function. 
    The number of distributions that do not overlap in total variation distance (i.e., are $\frac{1}{4}$-distinct), which we define as $\mathcal{C}$, can thus be lower-bounded:
    \begin{align}
        \mathcal{C} & \geq \frac{2^{N_y} / (N_y+1)}{2^{N_y \mathcal{H}_2 \left( \frac{1}{4} \right)}} \\
        & = \frac{2^{N_y(1  - \mathcal{H}_2(1/4))}}{N_y+1} \\
        & = \frac{2^{0.189 \times 2^{n_y}}}{2^{n_y} + 1}
    \end{align}
\end{proof}

The analysis of Theorem~\ref{thm:extra1} can now be applied to the second case, with a lower-bound of $\frac{2^{0.189 \times 2^{n_y}}}{2^{n_y} + 1}$ different circuits in place of $2^{2^{n_x}}$. 
\begin{lemma}
\label{thm:extra1-cor}
    $\tilde{C}_\textnormal{no input} \in \tilde{\Omega} (\exp(n_y))$.
\end{lemma}
\begin{proof}
Consider the minimal resources needed to prepare half of the different circuits (i.e., at least $  \frac{1}{2} \frac{2^{0.189 \times 2^{n_y}}}{2^{n_y} + 1} $ different circuits) using gates from some fixed gateset of size $n_g$, each of which acts on at most $n_d$ (qu)bits ($n_d$ would typically be 2 or 3 in practice). Therefore there are at most $(n_g n^{n_d})^{n_c}$ functionally distinct circuits of exactly $n_c$ gates, and at most  

\begin{equation}
    \sum_{i = 1}^{n_c} (n_g n^{n_d})^i \leq (n_g n^{n_d})^{n_c + 1}
\end{equation}
functionally distinct circuits of at most $n_c$ gates. 

We now wish to lower-bound $n_c$ needed to prepare $ \frac{1}{2} \frac{2^{0.189 \times 2^{n_y}}}{2^{n_y} + 1} $ different circuits:
\begin{align}
    (n_g n^{n_d})^{n_c + 1} \geq & \frac{1}{2} \frac{2^{0.189 \times 2^{n_y}}}{2^{n_y} + 1} \nonumber \\
    \implies (n_c + 1) \log_2 (n_g n^{n_d}) \geq & 0.189 \times 2^{n_y}  - 1 -\log_2 (2^{n_y} + 1)  \nonumber \\
    \implies n_c \geq & \frac{0.189 \times 2^{n_y}  - 1 -\log_2 (2^{n_y} + 1)}{\log_2 (n_g n^{n_d})} -1 \nonumber \\
\end{align}
and so when $n_d$ and $n_g$ are constants and $n \in \text{Poly}(n_\text{in}) = \text{Poly}(n_y)$, it is the case that to prepare $\frac{1}{2} \frac{2^{0.189 \times 2^{n_y}}}{2^{n_y} + 1}$ computational effects, $n_c \in \tilde{\Omega} (\exp(n_{y}))$. Thus, at least half of the computational effects must take more gates than this, and so $\tilde{C}_\text{no input} \in \tilde{\Omega} (\exp(n_{y}))$. 
\end{proof}

On the other hand, as and $n_\text{in} = n_q = \Theta (n_y)$ by definition, the analysis of Theorem~\ref{thm:extra2} can be applied again to obtain:

\begin{corollary}[Of Thm~\ref{thm:extra2}]
\label{thm:extra2-cor}
    Circuits exist such that  $\tilde{C}_{\textnormal{no input}} \in \tilde{\mathcal{O}} (\exp (n_{y}))$.
\end{corollary}

Thus, in the second case as well as the first, the number of different computational effects is simply too large for any resource circuit to make more than a negligible difference on average.

\section{Discussion}
\label{sec:discuss}

Though a little care and detail is needed to ensure the analysis leading up to it is rigorous, the independence theorem presented in this paper amounts to a very simple principle: the exhaustive list of computational effects that can be achieved with any given machine learning model structure is independent of any untrained resource quantum circuit almost surely. Hence, if there is no \textit{a priori} bias towards certain computational effects, then no resource quantum circuit is better than any other that can achieve the same set of computational effects. This conclusion is backed up by the general analysis in Section~\ref{sec:supp-analysis}, addressing the case where certain idealisations needed to make the no free theorem work are dropped.

This does not mean that untrained quantum circuits cannot offer some advantage in machine learning, for instance in the ways described references discussed in the introduction, but it does mean that there must be some reason why the untrained circuit is ``well matched'' to problem space in question. This could be the result of theoretical analysis or, more likely, demonstrated by empirical results. This, in turn, implies the importance of demonstrating the effectiveness of such proposals ``at scale'', as there is no reason to believe that any advantage observed in reduced-size proof-of-principle demonstrations will persist as the size increases.

Another important topic to discuss is whether the choice of cost, $c$, is appropriate. In particular in learning theory the effectiveness of machine learning models is generally measured in terms of expressivity, generalisation and trainability. As $c$ can be defined as monotonically decreasing with the number of computational effects that can be prepared within a certain number of gates -- i.e., a certain model capacity -- it can correspond directly to expressivity (as already noted and illustrated by item (ii) in the enumeration at the opening of Section~\ref{subsec:cost-target}). The other two measures (generalisation and trainability) are data-dependent, and so any independence theorem would entail averaging over every input dataset.

For the adjacent, but distinct problem of unitary processes from \textit{quantum} data, Sharma \textit{et al} show that there is indeed an overall ``resource value'' of entanglement in the training data set, in the sense that entanglement improves the fundamental learnability of quantum operations \cite{Sharma2022}. On the other hand, for our specific setting, one could argue that $c$ defined as the total gate count is a proxy for trainability and generalisation, by considering the following counterfactual: If it were the case that certain resource circuits significantly reduced the average number of gates needed to prepare every possible computational effect, then this would translate as smaller models being needed; this, in turn, would imply that training should be easier and overfitting less likely. With these two contradictory perspectives, we can only surmise that further analysis is needed to give a conclusive answer.

As a final discussion point, we return to the question of why one may ever believe that certain quantum circuits do possess superior resource value. Most evidently, one may suppose that highly-entangled highly ``non-classical'' quantum circuits / states would have superior resource value, as they would enable access to parts of the Hilbert space that are inaccessible by purely classical means. One may even go as far as to posit that something like entropy of entanglement \cite{Bennett1996} should measure the resource value. However, one can easily see that a sufficient condition for two $n_q$-qubit candidate resource states to prepare the same family of target distributions is that they differ by a computational basis state permutation (so when written as $2^{n_q}$ element vectors, the same values appear in each case, but with the elements they reside in permuted). As (for example) CNOT gates are both entangling and permutations it is easy to construct examples of a pair of states that differ by a computational basis state permutation where one state is a tensor-product state and the other is (potentially highly) entangled. \\

\section*{Acknowledgements}

Thanks to Alexandre Krajenbrink and Julien Sorci for reviewing a draft of this article and providing many valuable suggestions. Thanks also to the anonymous reviewers and editorial team.

%


\end{document}